\documentclass[11pt,oneside]{amsart}

\usepackage[utf8]{inputenc}
\usepackage[scale=0.65]{geometry}
\usepackage[colorlinks=true,linkcolor=blue]{hyperref}
\usepackage{microtype}
\usepackage{graphicx}
\usepackage{algorithm2e}

\newcommand{\aaa}{\mathbf a}
\newcommand{\rrr}{\mathbf r}
\newcommand{\ooo}{\mathbf o}
\newcommand{\ppp}{\mathbf p}

\newcommand{\saa}{\mathbf (sa)}
\newcommand{\maa}{\mathbf (ma)}
\newcommand{\laa}{\mathbf (la)}
\newcommand{\soo}{\mathbf (so)}
\newcommand{\moo}{\mathbf (mo)}
\newcommand{\loo}{\mathbf (lo)}

\DeclareMathOperator{\mina}{minang}
\DeclareMathOperator{\maxa}{maxang}

\newtheorem{thm}{Theorem}
\newtheorem{lem}[thm]{Lemma}
\newtheorem{prop}[thm]{Proposition}
\newtheorem{conj}[thm]{Conjecture}
\theoremstyle{remark}
\newtheorem*{remark}{Remark}

\title{Dissecting the square into seven or nine congruent parts}
\author[G. L. Maldonado]{Gerardo L. Maldonado}
\address[G. L. Maldonado]{Centro de Ciencias Matemáticas, UNAM Campus Morelia, Morelia, Mexico}
\email{gmaldonado@matmor.unam.mx}

\author[E. Roldán-Pensado]{Edgardo Roldán-Pensado}
\address[E. Roldán-Pensado]{Centro de Ciencias Matemáticas, UNAM Campus Morelia, Morelia, Mexico}
\email{e.roldan@im.unam.mx}


\begin{document}

\begin{abstract}
	We give a computer-based proof of the following fact: If a square is divided into seven or nine convex polygons, congruent among themselves, then the tiles are rectangles. This confirms a new case of a conjecture posed first by Yuan, Zamfirescu and Zamfirescu and later by Rao, Ren and Wang. Our method allows us to explore other variants of this question, for example, we also prove that no rectangle can be tiled by five or seven congruent non-rectangular polygons.
\end{abstract}

\maketitle

\section{Introduction}

Let $P$ and $T$ be convex polygons. We say that $P$ can be tiled by $n$ copies of $T$ if there are convex polygons $T_1\dots,T_n$, all congruent to $T$ such that $P=\bigcup T_i$ and the $T_i$ have disjoint interiors.

When can a polygon $P$ be tiled by $n$ copies of $T$? In this paper we are mainly interested in the case when $P$ is either a square or a rectangle.

It is easy to see that a square can always be tiled by $n$ congruent rectangles. This can be done by dividing the square by $n-1$ vertical lines. When $n$ is not prime, there are many other ways in which this can be done and tiles need not be constructed with vertical lines. However, it is not known if there is an odd number $n$ and a non-rectangular tile $T$ for which a square can be tiled by $n$ copies of $T$. The standing conjecture, as stated in \cite{RRW2020}, is as follows.

\begin{conj}\label{conj}
	If $n$ is an odd positive integer, then a square can be tiled by $n$ congruent copies of a convex polygon $T$ only if $T$ is a rectangle.
\end{conj}

This conjecture, for $n=3$, was posed as a problem by Rabinowitz in the journal \textit{Crux Mathematicorum} and was answered by Maltby \cite{Mal1991}. Maltby later generalized his result by showing that it is impossible to tile a rectangle by $3$ copies of $T$ unless $T$ is also a rectangle \cite{Mal1994}.

For $n=5$, Conjecture \ref{conj} was verified by Yuan \textit{et al.} \cite{YZZ2016}. They attribute a similar problem to Danzer, who conjectured that a square may not be tiled by $5$ congruent polygons (convex or not), except when $T$ is a rectangle. Danzer's conjecture remains open. They also posed Conjecture \ref{conj} for prime numbers $n$.

Apart from these two special cases, it is known that if a square can be tiled by an odd number of copies of $T$, then $T$ cannot be a triangle. This follows from the work of Thomas and Monsky \cite{Tho1968,Mon1970} who, independently, proved that a square cannot be tiled by an odd number of triangles with the same area. Recently Rao \textit{et al.} showed that $T$ may not have more than $6$ sides and that $T$ may not be a right-angle trapezoid \cite{RRW2020}.

In this paper we give a computer-based proof of the validity of Conjecture \ref{conj} for $n=7$ and $n=9$.

\begin{thm}\label{thm}
	Let $n=7$ or $n=9$, then a square cannot be tiled by $n$ copies of a convex polygon $T$ unless it is a rectangle.
\end{thm}

Using the same techniques, we are able to prove the following result for rectangles.

\begin{thm}\label{thm:rect}
	Let $n=5$ or $n=7$, then no rectangle can be tiled by $n$ copies of a convex polygon $T$ unless it is also a rectangle.
\end{thm}

Using slight modifications to our method, we are able to classify all tilings of the square using non-rectangular equiangular convex polygons. Here, two convex polygons are \emph{equiangular} if there is a bijection between their vertices, respecting the order of the vertices, such that the angles at corresponding vertices be equal.

\begin{thm}\label{thm:equiang}
	There are $31$ ways (in the sense described in Section \ref{sec:G}) in which a square can be tiled by $5$ non-rectangular equiangular convex polygons.
\end{thm}

Having this list, we are able to answer the first three open problems stated at the end of \cite{YZZ2016}.

The rest of the paper is devoted to proving these theorems. The main layout is as follows. We assume that a square or rectangle $P$ can be tiled by $n$ copies of a convex polygon $T$ which is not a rectangle. In Section \ref{sec:combi} we deduce properties that the tiles should have. In Section \ref{sec:G} we construct a polyhedral graph associated to a tiling of $P$. In Section \ref{sec:algorithm} we describe the algorithm we used to see which of these graphs could be obtained by a tiling of $P$ by congruent tiles. If there are no such graphs then the tiling does not exist. In Section \ref{sec:equiangSec} we mention the modifications needed to prove Theorem \ref{thm:equiang} and give several examples of such tilings. We give some final remarks in Section \ref{sec:remarks}.

The code we used is available at \url{https://github.com/XGEu2X/TilingSquare/}. The \verb|Readme.md| file in this repository contains a small explanation of how the code works. It is possible to print out parts of the proof to see the steps taken to reach our conclusions.

\section{Properties of tilings}\label{sec:combi}

We start by listing four properties that the tiles should have if they are to successfully tile a unit square. The first three properties are combinatorial and help to reduce the number of cases that must be analyzed later. The last property is geometric and is used later in the algorithm.

\begin{lem}\label{lem}
	Let $n\ge 3$ be an odd integer and let $P$ be a rectangle. Let $s_0,s_1,s_2,s_3$ be the (closed) sides of $P$ ordered cyclically, where the indices are taken mod $4$. If $P$ can be tiled by $n$ copies $T_1,\dots, T_n$ of a convex body $T$, then the following hold:
	\begin{enumerate}
		\item $T$ has at least $4$ sides.
		\item If a tile $T_i$ intersects two consecutive sides $s_k$ and $s_{k+1}$ of $P$, then $T_i$ contains the vertex of $P$ common to $s_k$ and $s_{k+1}$.
	\end{enumerate}
	Furthermore; if $P$ is a square, $T$ has exactly $4$ sides and is not a rectangle then:
	\begin{enumerate}
		\item[(3)] A tile $T_i$ cannot have two sides $a,b$ such that $a\subset s_k$ and $b\subset s_{k+2}$ for some $k$.
		\item[(4)] $T$ does not have two consecutive right angles.
	\end{enumerate}
\end{lem}

\begin{figure}
	\centering
	\includegraphics[scale=0.75]{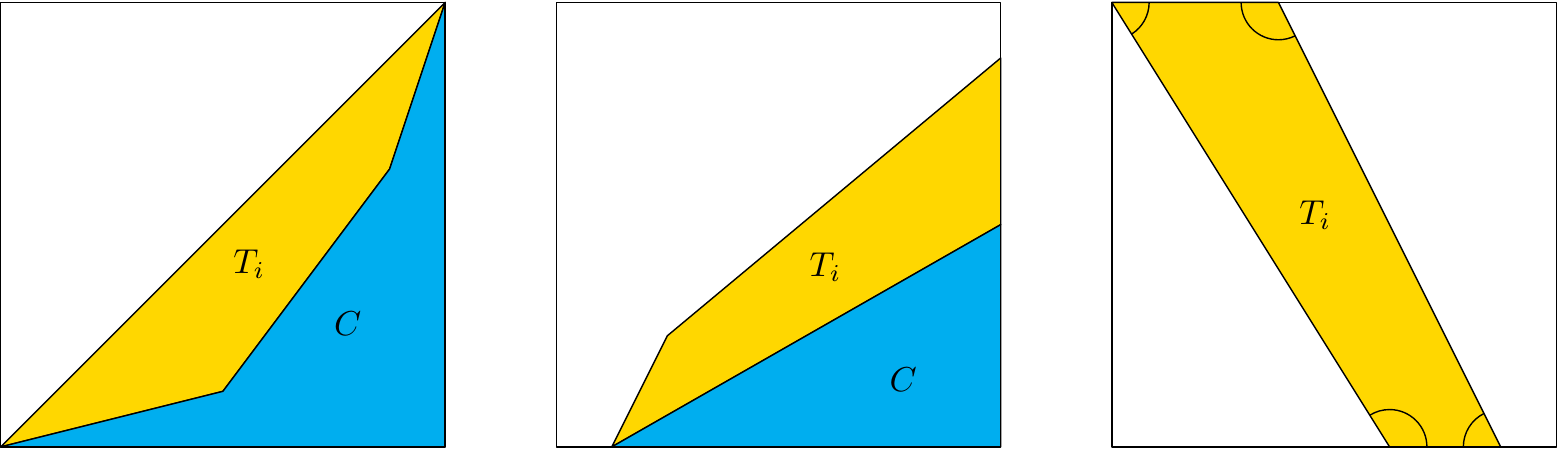}
	\caption{Three types of impossible tiles corresponding to parts (2) and (3) of Lemma \ref{lem}. In the first two pictures, the diameter of any convex body contained in $C$ is smaller than the diameter of $T_i$. In the last picture, the interior angles of $T_i$ are all greater than $\pi/4$.}
	\label{fig:lemma}
\end{figure}

\begin{proof}
	As mentioned before, if $P$ is a square then (1) follows immediately from the main results in \cite{Tho1968,Mon1970}. If $P$ is a rectangle then we can apply a linear transformation which sends it to a square. Since linear transformations preserve areas, the same result holds.
	
	Part (4) is Theorem 1.2 in \cite{RRW2020}.
	
	To prove (2), assume that $T_i$ intersects sides $s_k$ and $s_{k+1}$ but does not intersect the corner contained in $s_k\cap s_{k+1}$. There are two cases to consider here.
	If $T_i$ touches two opposite corners of $P$, as in Figure \ref{fig:lemma} (left), then the diameter of $T_i$ equals the diameter of $P$. Consider the connected components of $P\setminus T_i$: there are at least two of them or $T_i$ would cover half of $P$. One of these components must be properly contained in one of the right triangles obtained by splitting $P$ by its diagonal, but any tile $T_j$ contained in this component must have diameter smaller than $P$. This leads to a contradiction.
	The argument when $T_i$ does not touch two opposite corners of $P$, as in Figure \ref{fig:lemma} (middle), is similar. Let $T_j$ be a tile containing the corner of $P$ common to $s_k$ and $s_{k+1}$. Then the diameter of $T_j$ must be smaller than the diameter of $T_i$, leading to a contradiction.
	
	Finally, assume that (3) does not hold. We consider two cases, the first is when some corner of $P$ is incident with exactly one tile, then $T$ must have a right angle. Since $T_i$ has two parallel sides, it has two consecutive right angles which contradicts (4). The second case is when every corner of $P$ is incident with at least two tiles. Then $T$ must have an interior angle of at most $\pi/4$. But since $T_i$ has two vertices on $s_k$ and its other two vertices on $s_{k+2}$, all of its interior angles are strictly greater than $\pi/4$ which is a contradiction. This is exemplified in Figure \ref{fig:lemma} (right).
\end{proof}

\section{The graph associated to a tiling}\label{sec:G}

In this section we construct a polyhedral graph that contains the combinatorial structure of a tiling of the square.

Let $T_1,\dots,T_n$ be convex polygons which tile a square or rectangle $P$ with closed sides $S_1,S_2,S_3,S_4$ in cyclic order.

Construct a graph $G=(V,E)$, where $V=\{S_1,\dots,S_4,T_1,\dots,T_n\}$. The edges of this graph are defined by the following rules:
\begin{itemize}
	\item $\{A,B\}\in E$ if $A,B\in V$ and $A\cap B$ is a segment of positive length.
	\item $\{S_i,S_j\}\in E$ if $i-j\equiv 1\pmod{4}$.
\end{itemize}

$G$ contains all the combinatorial properties of the tiling.

We can think of $G$ as the dual of a pyramid with a square or rectangular base in which the base has been tiled in the same way as $P$. Furthermore, the following proposition gives us one more useful property.

\begin{prop}
	The graph $G$ is $3$-connected.
\end{prop}
\begin{proof}
	Let $v_1$,$v_2\in V$, each $v_i$ is either a tile or a side of $P$. We must prove that for every selection of these vertices, $G'=G-v_1-v_2$ is connected. It is enough to prove that, for each $T_i\in V(G')$ and each $S_j\in V(G')$, there is a path between $T_i$ and $S_j$. In this way, any two vertices of $G'$ may be connected by a path through some $S_j$ or some $T_i$. Let $p$ be a point in the interior of a $T_i\in V(G')$.
	
	If $v_1$ and $v_2$ are both sides of $P$, then choose a segment $l$ from $p$ to any point in the relative interior of $S_j\neq v_1,v_2$. The segment $l$ is contained in $P\setminus (v_1\cup v_2)$. By taking the graph induced by tiles that intersect $l$ (including $S_j$) we obtain a connected subgraph of $G'$ which contains $T_i$ and $S_j$. Therefore $G'$ is connected.
	
	Otherwise, consider a line $l$ separating the relative interiors of $v_1$ and $v_2$. Let $l'$ be the line parallel to $l$ that contains $p$. Since $v_1$ and $v_2$ are convex, the intersection between $l'$ and $(v_1\cup v_2)$ is a segment, a point or the empty set. In any case, there is a segment contained in $l'$ connecting $p$ and a side $S_j\neq v_1,v_2$ which does not intersect $v_1\cup v_2$. As above, this implies that $G'$ is connected.
\end{proof}

Steinitz's theorem \cite{Gru07} states that the family of polyhedral graphs (corresponding to convex polytopes) is precisely the family of $3$-connected planar graphs. Therefore, our graph $G$ is a polyhedral graph.

\begin{figure}
	\includegraphics[scale=0.7]{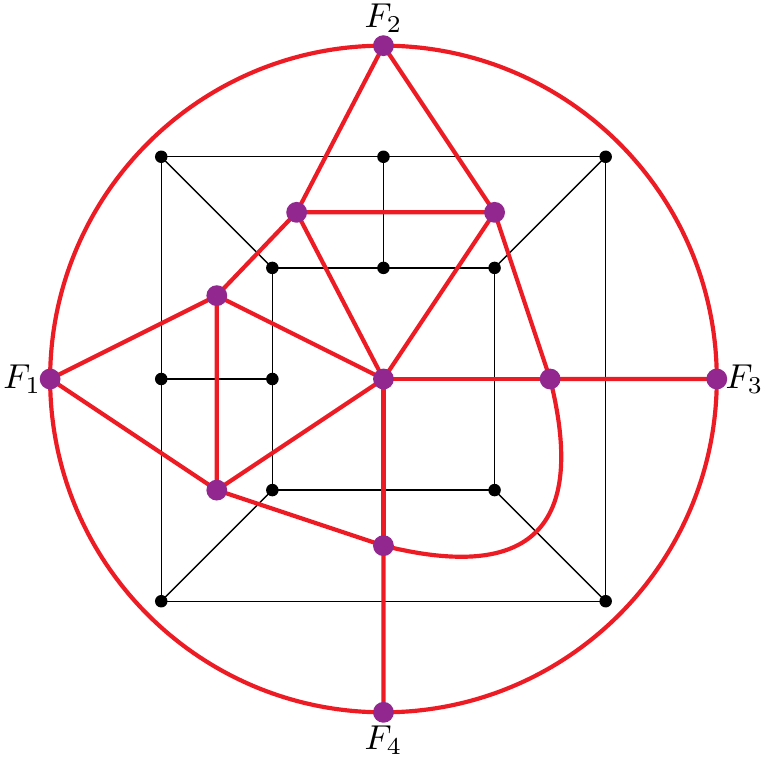}
	\caption{A tiling (black) and its associated graph (red).}
	\label{fig:graphG}
\end{figure}

Note that the degree of a tile $T$ in $G$ may not correspond to the number of sides of $T$ as a polygon. However, the number of sides of $T$ is a lower bound for the degree of $T$ in $G$.

In this way we have associated a graph to each convex tiling of the square by tiles of at least $4$ sides. This graph has a distinguished set of vertices $S$ which form a cycle. Figure \ref{fig:graphG} shows an example of a tiling with its associated graph $G$.

In what follows, we study pairs $(G,S)$ where $G=(V,E)$ is a $3$-connected planar graph with $n+4$ vertices, and $S\subset V$ induces a $4$-cycle in $G$. Since these graphs come from tilings of the square, we simply refer to the vertices of $G$ in $S$ as \emph{sides} and the vertices in $V\setminus S$ as \emph{tiles}. We say that two pairs $(G,S)$ and $(G',S')$ are isomorphic if there is an isomorphism $\phi$ between $G$ and $G'$ such that $\phi(S)=S'$.

To avoid confusion in the future, the union of the sets of vertices of the polygons $T_i$ are called \emph{tiling-vertices}. The set of tiling-vertices in the boundary of a tile $T_i$ are called the \emph{tiling-vertices of $T_i$}; the tile $T_i$ has internal angles at each of these, although some may be equal to $\pi$.

\begin{remark}
	In practice, we need to generate all such graphs using \verb|plantri| \cite{BMo2007}. To facilitate this, we add an extra vertex $S_0$ to $G$ which is adjacent to the vertices in $S$. Since we deal mostly with quadrilateral tiles, we can specify that this new graph has minimum degree $4$. It also simplifies checking isomorphisms, since we only need this new vertex to be fixed.
\end{remark}

\section{Exploring the tilings}\label{sec:algorithm}

Let $P$ be either a rectangle with sides of length $1$ and $r$, or the unit square (in which case we take $r=1$ as a constant). Let $S$ be the set of sides of $P$. Assume that $P$ can be tiled by copies $T_1,\dots, T_n$ of a convex body $T$ which is not a rectangle. By part (1) of Lemma \ref{lem}, $T$ has at least $4$ sides. Therefore, as described in Section \ref{sec:G}, there is a pair $(G,S)$ associated to this tiling where $G$ is a $3$-connected planar graph on $n+4$ vertices such that its tiles have minimum degree $4$.

Using \verb|plantri| it is possible to generate all graphs $G$ with these properties with $n\le 9$.
Once this is done, we search each graph $G$ to find the possible distinguished $4$-cycles. A single graph may have several possible distinguished $4$-cycles, so we use the isomorphism algorithm in \verb|NetworkX| \cite{HSS2008} to avoid including two pairs $(G,S)$ which are isomorphic.

Now we filter these lists by using part (2) of Lemma \ref{lem}. Since this is a purely combinatorial statement, it is easy to check it directly on $(G,S)$.

At this point, the number of sides of $T$ becomes relevant. We split the analysis into cases. If $T$ is a quadrilateral then we use part (3) of Lemma \ref{lem} to filter the list once more.
If not, then the tiles in $(G,S)$ necessarily have degree at least $5$, so we may discard the graphs that do not satisfy this.

In order to filter this list further, we must use geometrical properties of $T$ such as the values of its angles and side-lengths.
Let $p_1,p_2,\dots,p_k$ be the vertices of $T$, $\alpha_i$ the internal angle of $T$ at $p_i$ and $t_i$ be the length of the segment $p_ip_{i+1}$. The angles must satisfy the equation $\alpha_1+\alpha_2+\dots+\alpha_k=(k-2)\pi$.

Recall that if $v$ is a tiling-vertex of a tile $T_i$, then the internal angle of $T_i$ at $v$ is either $\pi$ or one of the $\alpha_j$. For every tiling-vertex $v$, the sum of the internal angles at $v$ of the tiles that contain $v$ must be either $2\pi$, $\pi$ or $\pi/2$ depending on whether $v$ is in the interior, on a side or on a vertex of $P$. This is shown in Figure \ref{fig:eq}. Therefore, the angles $\alpha_1,\alpha_2,\dots,\alpha_k$ must satisfy a system of non-homogeneous linear equations.

\begin{figure}
	\includegraphics[scale=0.7]{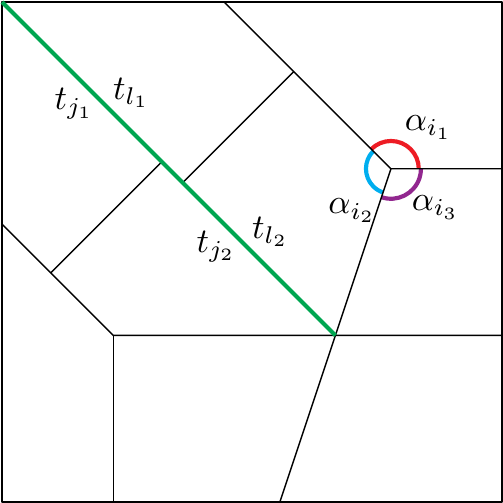}
	\caption{In this tiling we see that $\alpha_{i_1}+\alpha_{i_2}+\alpha_{i_3}=2\pi$, and $t_{j_1}+t_{j_2}=t_{l_1}+t_{l_2}$.}
	\label{fig:eq}
\end{figure}

Something similar happens for the side-lengths of $T$. For every side $s$ of $P$, take the tiles $T_i$ adjacent to $s$. The sum of the sides of these tiles contained in $s$ must add up to the length of $s$ which is either $1$ or $r$. We do not fix the value of $r$ but treat it as a variable. So this gives us four non-homogeneous linear equations for the $t_i$ and, if applicable, $r$.

Furthermore, if two tiles $T_i$ and $T_j$ intersect in a segment $v_1v_2$ such that $v_1$ and $v_2$ are consecutive corners of both $T_i$ and $T_j$, then we can deduce an equation of the form $t_k=t_l$.
More generally, if there is a line $\ell$ which contains a side of the tiles $T_{i_1},\dots,T_{i_{l}}$ and $T_{j_1},\dots,T_{j_{l'}}$ such that the $T_{i_m}$ and $T_{j_{m'}}$ are on different sides of $\ell$ and the $T_{i_m}$ cover the same segment of $\ell$ as the $T_{j_{m'}}$, then we may deduce another linear equation for the $t_i$ (see Figure \ref{fig:eq}).

In the quadrilateral case we may also check non-linear equations which involve both sides and angles of $T$. The area of $T$ is $r/n$, so \[t_1t_2\sin(\alpha_2)+t_3t_4\sin(\alpha_4)=t_2t_3\sin(\alpha_3)+t_4t_1\sin(\alpha_1)=2r/n\]
and by computing the length of the diagonals of $T$ we obtain
\begin{align*}
t_1^2+t_2^2-2t_1t_2\cos(\alpha_2)&=t_3^2+t_4^2-2t_3t_4\cos(\alpha_4)\\
t_2^2+t_3^2-2t_2t_3\cos(\alpha_3)&=t_4^2+t_1^2-2t_4t_1\cos(\alpha_1).
\end{align*}

All of these equations must be satisfied.
The problem is that, from the graph $G$, we do not know which tiling-vertices of a tile $T_i$ correspond to which vertices of $T$. If the degree of $T_i$ in $G$ is larger than $4$ then we must also decide which tiling-vertices of $T_i$ have internal angles equal to $\pi$. We could try computing all possibilities but the number of cases is too large even for a single graph, so instead we do something slightly more efficient which allows us to discard most of the remaining graphs.

Each angle of $T$ is labeled by $\aaa,\rrr$ or $\ooo$ depending on whether the angle is acute, right or obtuse. We also label the angles of each tile $T_i$ at each tiling-vertex of $T_i$ with $\aaa,\rrr,\ooo,\ppp$ depending on whether the angle is acute, right, obtuse or plain (meaning that the internal angle at this point is $\pi$). We call the label of an angle its \emph{angle-type}. We use the fact that the angles of each $T_i$ that are not labeled with $\ppp$ must appear in the same cyclic order as the angles of $T$.

There are several ways in which the internal angles of $T$ can be labeled. This depends on whether $P$ is a square or a rectangle, but they are easy to list. For example, if $T$ is a quadrilateral, the internal angles of $T$, ordered cyclically and without taking orientation into account, can be labeled in exactly $9$ different ways: $\aaa\aaa\aaa\ooo$, $\aaa\aaa\rrr\ooo$, $\aaa\aaa\ooo\ooo$, $\aaa\rrr\aaa\ooo$, $\aaa\rrr\ooo\rrr$, $\aaa\rrr\ooo\ooo$, $\aaa\ooo\aaa\ooo$, $\aaa\ooo\rrr\ooo$ and $\aaa\ooo\ooo\ooo$. In this case we are using property (4) of Lemma \ref{lem} to discard labels with two consecutive right angles. Since $T$ is not a rectangle, we do not include $\rrr\rrr\rrr\rrr$. If $T$ has $5$ sides, the number of ways to label the angles is greater.

There are several conditions which the angle-types of the tiles must satisfy. For example, assume that $v_k$ is a tiling-vertex of $T_i$. If $v_k$ is also a corner of $P$, then the angle-type of $T_i$ at $v_k$ is either $\aaa$ or $\rrr$, depending on whether there are other tiles containing $v_k$ or not.
To state another example, assume that $T_i$ and $T_j$ are the only two tiles containing a tiling-vertex $v_k$ which lies on a side of $P$. Then either the angle-types of $T_i$ and $T_j$ at $v_k$ are both $\rrr$ or one is $\aaa$ and the other is $\ooo$. See Figure \ref{fig:angtype} for an example of a tiling with some angle-types.

\begin{figure}
	\includegraphics[scale=0.7]{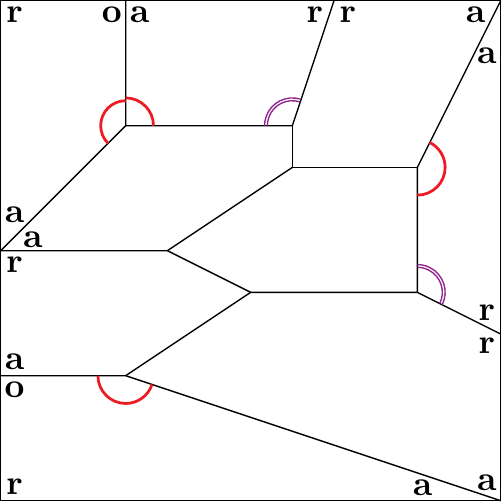}
	\caption{This is a representation of a tiling where the angles of $T$ are labeled as $\aaa\rrr\ooo\rrr$. At this point we are only working with the graph $G$, so the angles in this representation may not correspond to the actual angles of the tiles and some angles may even be plain.
	In this case there is only one possibility for the angle-types of the angles that lie in the boundary of $P$. From there it follows that the red angles must be labeled with $\rrr$ and the purple angles with $\ooo$.}
	\label{fig:angtype}
\end{figure}

In order to avoid listing these properties individually, we introduce the following definitions.
Let $\varepsilon>0$ be a small real number. Define $\mina(T_i,v_k)$ as $\varepsilon$, $\pi/2$, $\pi/2+\varepsilon$ or $\pi$ depending on whether the angle-type of $T_i$ at $v_k$ is $\aaa$, $\rrr$, $\ooo$, or $\ppp$, respectively. Likewise, define $\maxa(T_i,v_k)$ as $\pi/2-\varepsilon$, $\pi/2$, $\pi-\varepsilon$ or $\pi$ depending on whether the angle-type of $T_i$ at $v_k$ is $\aaa$, $\rrr$, $\ooo$, or $\ppp$. If $\varepsilon$ is small enough, then values of the angles of $T_i$ at $v_k$ are in the interval $[\mina(T_i,v_k),\maxa(T_i,v_k)]$.
Thus, for every tiling-vertex $v_k$ we have that
\begin{align*}
\sum_{T_i\ni v_k}&\maxa(T_i,v_k)\ge \alpha(v_k),\\
\sum_{T_i\ni v_k}&\mina(T_i,v_k)\le \alpha(v_k),
\end{align*}
where $\alpha(v_k)$ is $\pi/2$, $\pi$ or $2\pi$ if $v_k$ is in a corner, side or interior of $P$, respectively.
In practice, this allows us to decide the angle-types of the tiles fairly quickly, even if we do not know the values of the angles of $T$. Another useful observation comes from noting that each of the previous sums has at most $n$ terms; since we only work with $n\le 9$, any value of $\varepsilon$ smaller than $\pi/18$ allows us to discriminate between valid and invalid assignations of angle-types.

Now we are ready to filter the list of pairs $(G,S)$. For each of these pairs, select one of the possible angle-type assignations for the internal angles of $T$. For each tile $T_i$, we wish to decide which tiling-vertices of $T_i$ correspond to which vertices of $T$. So we assign to $(T_i,v_k)$, where $v_k$ is a tiling-vertex of $T_i$, a vertex of $T$ or a mark indicating that the angle-type of $T_i$ at $v_k$ is $\ppp$.
In order to do this, we use a deep search algorithm to explore the tree of possibilities. The process is divided into three main steps:
\begin{description}
	\item[Selection] Select a tile $T_i$ and one of its tiling-vertices $v_k$ such that $(T_i,v_k)$ is unassigned.
	\item[Assignation] Makes a list of valid assignations for $(T_i,v_k)$ taking into account the conditions that the angle-types must satisfy, described above.
	\item[Equation verification] For each assignation, check if a new equation for the angles or the sides of $T$ is generated. If so, verify that the equations still have solutions. This is done with \verb|Sympy| \cite{MSP+2017} which works symbolically instead of numerically, this guarantees that we do not discard equations when solutions actually do exist.
	If possible, check that the non-linear equations are not violated. This final step is done numerically, so a small tolerance is allowed.
\end{description}
The surviving assignations are added to the stack for further exploration.

The order in which the tiles and tiling-vertices are explored is important. We always start with the corners of the rectangle and the tiles that contain them. We continue with tiles adjacent to a side of $P$ and tiling-vertices in this side in a cyclic order. In many cases, exploring these pairs of tiles and tiling-vertices is enough to discard $(G,S)$ for a given angle-type labeling of $T$.

This does not eliminate all possible pairs $(G,S)$, so we also check the additional geometric conditions. These are described in the following statements.

\begin{lem}\label{lem:geom}
	Let $T$ be a quadrilateral with area $A$ and let $\alpha_i$ and $t_i$ be as above. If $T$ is labeled with angle-types $\aaa\rrr\ooo\rrr$ and $\alpha_1=\pi/4$, as in the left part of Figure \ref{fig:lgeom}, then
	\[\sqrt{2A}< t_1,t_4 \le 2\sqrt{A}.\]
	If instead $T$ is labeled with angle-types $\aaa\rrr\rrr\ooo$, as in the right part of Figure \ref{fig:lgeom}, then
	\begin{align*}
		t_1 &= t_3+t_4\cos(\alpha_1),\\
		t_2 &= t_4\sin(\alpha_1).
	\end{align*}
\end{lem}

\begin{figure}
	\includegraphics[scale=0.7]{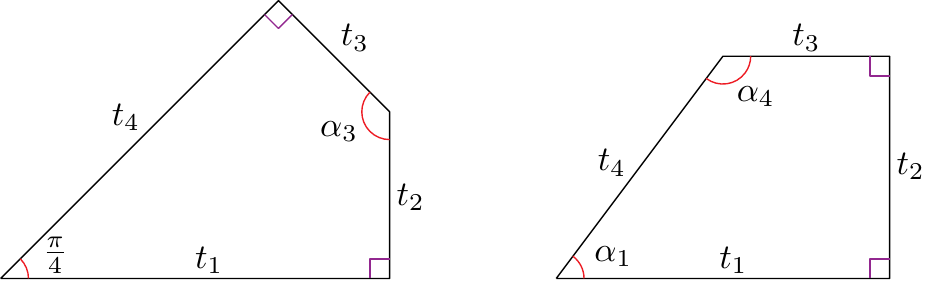}
	\caption{The quadrilaterals in Lemma \ref{lem:geom}}\label{fig:lgeom}
\end{figure}

\begin{proof}
	It is easy to see that $t_4=(t_1+t_2)/\sqrt{2}$ and $t_3=(t_1-t_2)/\sqrt{2}$, therefore
	\[t_1t_2+\frac{t_1^2-t_2^2}{2}=2A.\]
	Thinking of $t_1$ as a function of $t_2$ and taking the derivative with respect to $t_2$ yields
	\[t_1'=\frac{t_2-t_1}{t_2+t_1}.\]
	Since $\alpha_1=\pi/4$, we have $0<t_2<t_1$, so $t_1'$ is negative and the extremal values for $t_1$ occur when $t_2=0$ and $t_2=t_1$ which gives the desired result.
	
	The second part is straightforward.
\end{proof}

\begin{lem}\label{lem:aror}
	Let $T$ be a quadrilateral which tiles $P$ and has two angles labeled as $\rrr$. Assume some tile $T_i$ has a vertex on corner of $P$ with angle-type $\aaa$, then the acute angle of $T$ must be $\pi/3$ or $\pi/4$.
\end{lem}
\begin{proof}
	Let $\alpha_a$ and $\alpha_o$ be the values of the acute and obtuse angles of $T$. In the tiling, the number of obtuse angles must be equal to the number of acute angles. Note that no obtuse angle coincides with a corner of $P$ and any obtuse angle which coincides with a side of $P$ shares a vertex with exactly one acute angle of a different tile.
	Since there is an acute angle which coincides with a corner, there exists a tiling-vertex $v$ in the interior of $P$ incident to more obtuse angles than acute angles. Assume that the number of acute, obtuse and right angles incident to $v$ are $N_a$, $N_o$ and $N_r$, respectively. The sum of angles in each interior tiling-vertex is $2\pi$, so $N_o\alpha_o+N_a\alpha_a+N_r\pi/2=(N_o-N_a)\alpha_o + N_a\pi+N_r\pi/2=2\pi$. Since $\alpha_o>\pi/2$, we have that $N_o-N_a\le 3$, so $\alpha_o$ is of the form $p\pi/(2q)$ with $1\le p\le 4$ and $1\le q\le 3$. From here we may conclude that the only possible values for $\alpha_o$ are $2\pi/3$ and $3\pi/4$ and therefore $\alpha_a$ is $\pi/3$ or $\pi/4$.
\end{proof}

When $P$ is a square and $n\le 9$ is an odd integer or when $P$ is a rectangle and $n\le 7$ is an odd integer, every pair $(G,S)$ is discarded with this method. This proves Theorems \ref{thm} and \ref{thm:rect}.

\section{Equiangular pieces}\label{sec:equiangSec}

Using the same method, we explore another variant of the conjecture. At the end of \cite{YZZ2016} there are four open problems:
\begin{enumerate}
	\item Does every dissection of the square into five similar convex tiles use right isosceles triangles or rectangles as tiles?
	\item Does every dissection of the square into five equiangular convex polygons use only angles measuring $\pi/4$, $\pi/2$, $3\pi/4$?
	\item Find all dissections of the square into five equiangular non-rectangular convex polygons.
	\item Is every dissection of the square into $n$ congruent convex tiles necessarily the ``standard'' one (i.e. dividing it by $n-1$ vertical or horizontal lines) if $n\ge3$ is a prime number?
\end{enumerate}

Our only contribution to the fourth problem is that of Theorem \ref{thm}. However, we are able to solve the other three.
The first two questions have a negative answer, this can be seen immediately after solving the third problem.

By slightly modifying the algorithm described in Section \ref{sec:algorithm}, we are able to list all the possible equiangular tilings (in the sense described in Section \ref{sec:G}). These are shown in Figures \ref{fig:5_3equiang} and \ref{fig:5_4equiang}.

The equivalence relation we use might not be the one one might expect.
For example, in Figure \ref{fig:5_4equiang} there are four pairs of tilings enclosed in dashed rectangles. The two tilings in each rectangle have the same pair $(G,S)$, but one might argue that they are geometrically distinct.

The $31$ ways referred to in Theorem \ref{thm:equiang} can be separated into two types of tiling; then ones in which the tiles are triangles and the ones in which they are quadrilaterals.

\begin{figure}
	\includegraphics[scale=0.7]{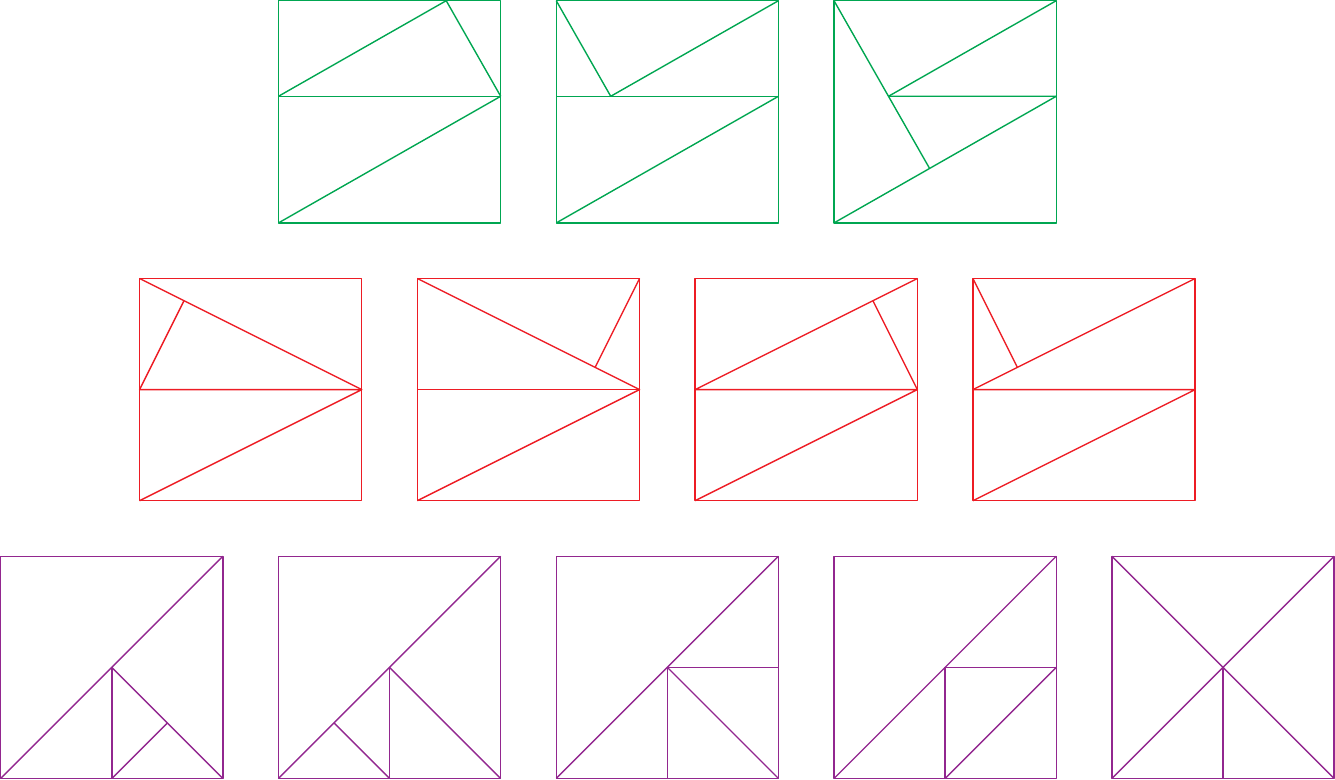}  
	\caption{Equiangular tilings with five triangles.}\label{fig:5_3equiang}
\end{figure}

Figure \ref{fig:5_3equiang} shows the $12$ ways in which a square can be tiled using $5$ similar triangles. The tiles in each case are right triangles, let $\alpha$ be the smallest angle of the triangles in each case. The top three tilings in Figure \ref{fig:5_3equiang} (green) have $\tan(\alpha) = a \approx 0.56984$, where $a$ is the real root of the polynomial $a^3-a^2+2a-1$.
The tilings in the middle row (red) have $\tan(\alpha)=1/2$. The bottom tilings (purple) have $\tan(\alpha)=1$.

\begin{figure}
	\includegraphics[scale=0.7]{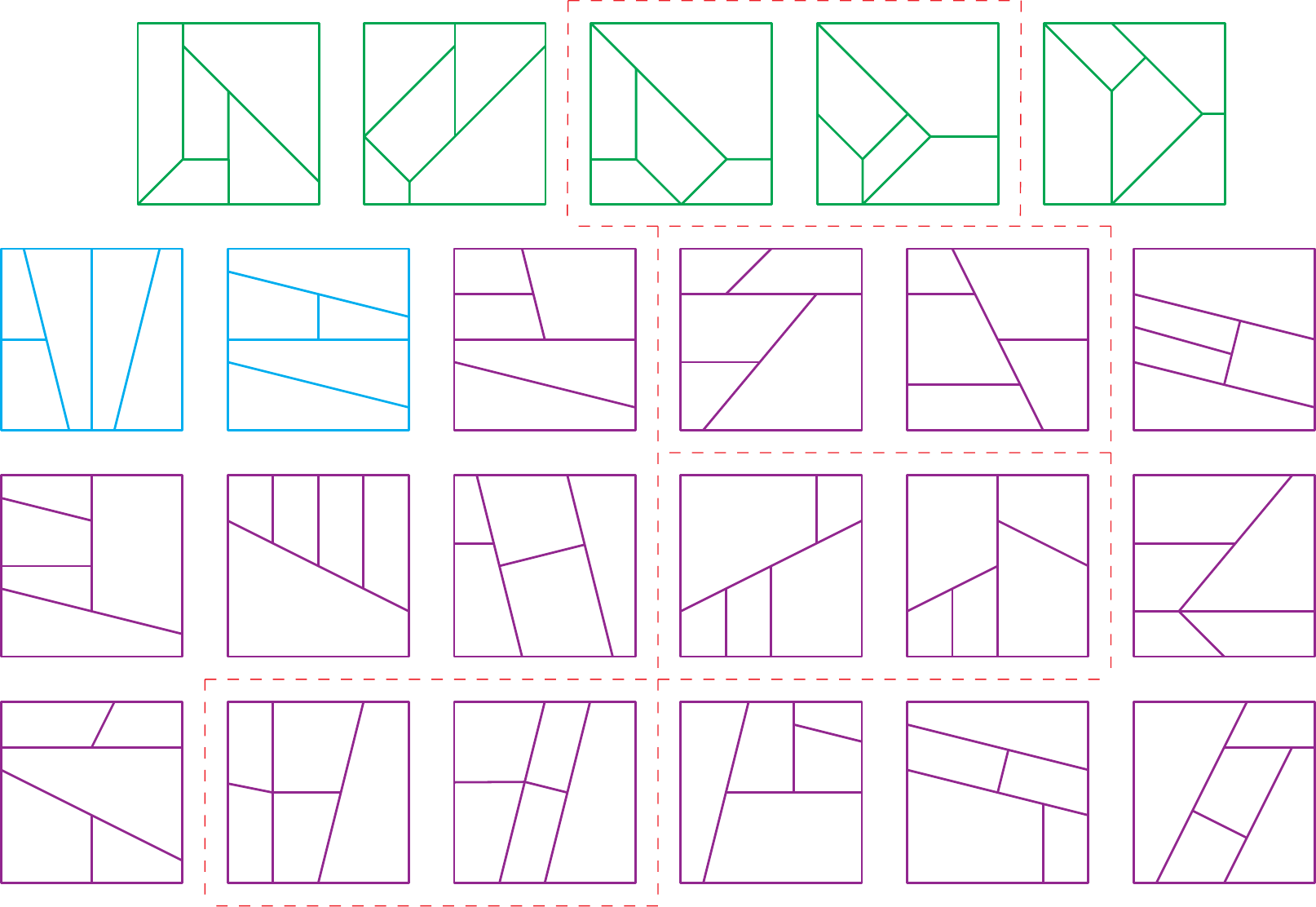}  
	\caption{Equiangular tilings using five quadrilaterals. The angles of the tiles are of the form $\alpha,\pi/2,\pi/2,\pi-\alpha$.}\label{fig:5_4equiang}
\end{figure}

The quadrilateral case is shown in Figure \ref{fig:5_4equiang}. Here each pair $(G,S)$ can be realized as a continuum of tilings as there are always several degrees of freedom. As mentioned before, the tiles which are enclosed in dashed rectangles are equivalent and are included only to illustrate what the equivalence means. In total, there are $19$ families of tilings of the square into $5$ equiangular quadrilaterals such that the elements of each class produce the same pair $(G,S)$. The families represented in the first row of Figure \ref{fig:5_4equiang} (green) have angles $\pi/4,\pi/2,\pi/2,3\pi/4$. The rest have angles $\alpha,\pi/2,\pi/2,\pi-\alpha$, where $\alpha$ can be chosen, in all but two tilings, in the interval $(0,\pi/4)$ and can sometimes take more values. The first two tilings (blue) are the exception; $\alpha$ can only be taken in the interval $(0,\arctan(1/2))$.

In order to prove Theorem \ref{thm:equiang}, we slightly modify the algorithm described in Section \ref{sec:algorithm}.

The obvious change is to ignore the side-lengths, so the equations only involve angles.
Another helpful change, which greatly improves the running time of our algorithm, is to use more angle-types when tiling with triangles. In this case, instead of using $\aaa$ to label an acute angle, we use the labels $\saa$, $\maa$ and $\laa$ which correspond to acute angles smaller than $\pi/4$, equal to $\pi/4$ and greater than $\pi/4$, respectively. Analogously, we replace $\ooo$ with $\soo$, $\moo$ and $\loo$ which correspond to obtuse angles smaller than $3\pi/4$, equal to $3\pi/4$ and greater than $3\pi/4$, respectively. Since Lemma \ref{lem} does not apply in this case, we use the following observations.

\begin{lem}
	Let $n\ge 3$ be an odd integer and $P$ be a square. Let $s_1,s_2,s_3,s_4$ be the (closed) sides of $P$ ordered cyclically, where the indices are taken mod $4$. If $P$ can be tiled by convex polygons $T_1,\dots, T_n$ which are equiangular and the $T_i$ are triangles, then the following hold:
	\begin{enumerate}
		\item If a tile $T_i$ intersects a side $s_k$ in a segment of positive length and one of its vertices $v$ is in $s_{k+2}$, then the angle $\alpha$ of $T_i$ at $v$ satisfies $\alpha<\pi/2$.
		\item If a tile $T_i$ contains a side $s_k$ and intersects $s_{k+1}$ (or $s_{k-1}$) in a segment of positive length, then the angle $\alpha$ at the vertex of $T_i$ in $s_{k-1}$ (or $s_{k+1}$) of $T_i$ satisfies $\alpha\leq\pi/4$.
		\item A tile $T_i$ cannot have two sides $a,b$ such that $a\subset s_k$ and $b\subset s_{k+2}$ for some $k$.
	\end{enumerate}
	If instead the $T_i$ are quadrilaterals, then
	\begin{enumerate}
		\item[(4)] A tile intersecting three consecutive sides of $P$ in segments of positive length, cannot have an angle $\alpha \leq \pi/4$.
	\end{enumerate}
\end{lem}
\begin{proof}
	In (1), the two sides adjacent to $v$ have length at least $1$, so the side in $s_{k}$ is the smallest side of $T_i$. Thus, the smallest angle of $T_i$ is $\alpha$ and therefore $\alpha\le\pi/3<\pi/2$.
    Statement (2) is straightforward from the fact that $T_i$ is a right triangle and $\alpha$ is its smallest angle.
    Statement (3) follows from the fact that, in a triangle, every side is adjacent to the other two.
    Finally, in (4), suppose $T_i$ intersects $s_k$, $s_{k+1}$ and $s_{k+2}$. Since $T_i$ is a quadrilateral, $s_{k+1}$ must be one of its sides and the vertices of $T_i$ which are not in $s_{k+1}$ are in $s_k$ and $s_{k+2}$. Therefore, the two angles incident to $s_{k+1}$ are right angles and each of the other two are greater than $\pi/4$.
\end{proof}

After running the modified program, we are left with $15$ graphs in the triangular case and $27$ graphs in the quadrilateral case. Of these, $3$ graphs for the triangular case and $8$ graphs for the quadrilateral do not produce valid tilings. This can easily be checked by hand. All these graphs can be seen in our repository, it is also possible to examine the unrealizable and unfiltered graphs.

\section{Final remarks}\label{sec:remarks}

We decided to use Python in order to make the code more accessible. The execution time turned out not to be an issue. Since each graph can be explored separately, the code is easily parallelizable. Proving Theorem \ref{thm} for $n=9$ takes a couple of days in a modern home computer, the rest of cases take at most a few of hours.

In Theorems \ref{thm} and \ref{thm:rect}, the number of tiles may not be increased using our method. The number of graphs to be analyzed is simply too large. As for the equiangular case with $7$ tiles, our algorithm produces around $2000$ possible valid graphs, however we were unable to find a systematic way of deciding whether a graph is valid or not.

Notice that this method can be used to find tilings (using congruent or equiangular pieces) of other polygons. We may use other observations in the way of Lemma \ref{lem} to optimize the process. It might be necessary to consider labellings using different angle-types. What is important is to check that the underlying graph is $3$-connected or that there is some other condition which allows us to generate the list of possible graphs.

Returning to the general conjecture, we noticed that when the angles of $T$ are (cyclically) labeled as $\aaa\rrr\ooo\rrr$, we usually had to go deeper in the tree of possibilities before we could completely discard a given pair $(G,S)$. We tried to rule out this case in general but failed. It seems that the corresponding tiles have more structure so it may be possible to discard them by other means.  

\section{Acknowledgments}

We are thankful to the anonymous referee for his careful reading and fruitful suggestions.
During this research the first author was supported by a CONACyT scholarship and the second author was supported by CONACyT project 282280.

\bibliographystyle{amsalpha}
\bibliography{biblio}

\end{document}